\documentclass[titlepage,12pt]{article}
\usepackage{amsmath,amsthm,amssymb}
\usepackage[mathscr]{eucal}
\usepackage[latin1]{inputenc}
\usepackage{graphicx} 
\usepackage{url}     
\usepackage{natbib}

\newtheorem{The}{Theorem}
\newtheorem{Pro}[The]{Proposition}

\theoremstyle{definition}

\numberwithin{equation}{section} \numberwithin{The}{section}

\begin{document}

\title{Assessing dominance in survival functions: A test for right-censored data}
\author{F\'{e}lix Belzunce, Carolina Mart\'{\i}nez-Riquelme and Jaime Valenciano \\ Dpto. Estad\'{\i}stica e Investigaci\'{o}n Operativa \\ Universidad de Murcia \\ Facultad de Matem\'{a}ticas, Campus de Espinardo \\ 30100 Espinardo (Murcia), SPAIN \\ \texttt{belzunce@um.es, carolina.martinez7@um.es, jaime.valencianog@um.es} }
\date{April, 2025}
\maketitle

\maketitle

\begin{abstract} This paper proposes a new statistical test to assess the dominance of survival functions in the presence of right-censored data. Traditional methods, such as the log-rank test, are inadequate for determining whether one survival function consistently dominates another, especially when survival curves cross. The proposed test is based on the supremum of the difference between Kaplan-Meier estimators and allows for distinguishing between dominance and crossing survival curves. The paper presents the test?s asymptotic properties, along with simulations and applications to real datasets. The results demonstrate that the test has high sensitivity for detecting crossings and dominance compared to conventional methods. 


\end{abstract}

\section{Introduction}
Longevity has gained immense importance in modern society due to demographic, social, economic, and scientific factors. As lifespans increase worldwide, the focus shifts not only to living longer but also to living healthier. This applies not only to healthy individuals but also to those undergoing treatment for illness. In such cases, analyzing time-to-death---or organ dysfunction---data under different conditions provides essential insights.

Survival analysis is a branch of statistics that focuses on the analysis of time-to-event data. A common challenge in this field arises when data are right-censored, meaning that for some subjects, the event of interest (e.g., death, failure, or relapse) has not occurred by the end of the observation period. For instance, in clinical studies on lung cancer, some patients may still be alive at the end of the study, meaning that their exact time of death is unknown. In such cases, their survival time is considered right-censored, as it is only known to exceed the duration of follow-up. Examples of censored data appear in sections 2 and 3. Since traditional statistical methods are unsuitable for handling censored data, specialized techniques are required to compare survival functions between groups. For a review of these techniques, see \cite{KleinMoeschberger:2003}.

Recall, that the survival function, $S(t)$, represents the probability that the event of interest occurs later than time $t$:

$$
S(t) = P(T > t),
$$
where $T$ is the time to the event. When one survival function lies consistently below (or above) another, it indicates that the probability of surviving beyond any time $t$ is always lower (or higher) for one group compared to the other. From a probabilistic perspective, this implies that the random lifetime is generally smaller (or larger) for one group than for the other. Detecting whether this situations occurs -or whether there is at least one crossing point among the survival functions- is crucial for identifying the factors or conditions associated with longer lifetimes. 

The primary approach for performing such an analysis using sample observations involves comparing non-parametric estimators of the survival functions and applying a significance test to determine whether they can be considered equal or not. When right-censoring occurs the Kaplan-Meier estimator is a widely used non-parametric method for estimating the survival function. Additionally, several statistical tests are available to assess whether survival functions differ.

The log-rank test is the most commonly used statistical test for comparing survival curves. It compares the observed and expected number of events at each time point across groups under the null hypothesis that the survival functions are identical. Other tests include: 
\begin{enumerate}

\item The Gehan-Wilcoxon test, which assigns more weight to early time points in the survival curve, making it more sensitive to differences in survival that occur earlier in the study period.

\item The Tarone-Ware test, which is a compromise between the log-rank and Wilcoxon tests.

\item Peto-Peto-Prentice test, an adaptation of the log-rank test that also places more emphasis on earlier time points.

\end{enumerate}

All these tests compare survival only at observed time points, assuming the null hypothesis that survival functions are equal. For tests that assess survival differences over an interval, we have the approach proposed by \cite{Gill:1980}. However,  traditional methods for detecting differences between survival functions are not suitable for addressing the aforementioned problem, do not determine whether one survival curve dominates the other, and are inefficient when survival curves cross.

This work proposes a new test based on the supremum of the difference between Kaplan-Meier estimators over an interval, capable of detecting differences and distinguishing whether the survival functions exhibit ordering or crossings.

The organization of this paper is as follows. In Section 2, we introduce a new test for assessing dominance between two survival functions in the presence of censored data. We examine its asymptotic properties and consistency. Next, in Section 3, we show the implementation of the new test and apply to real datasets.  In Section 4, we conduct a simulation study to evaluate its performance under different scenarios. Throughout the paper, we assume that the samples are independent. Finally, in Section 5 we provide a discussion on advantages and limitations.

\section{Test for comparing survival functions. The case of independent samples}

As mentioned in the introduction, in this paper we consider right-censored data, that is, in some cases we know the time at which some event occurs (e.g., failure, death), but in some others we know that the event has not yet happened by the end of the observation period, but we do not know the exact time at which it will occur.

Let us consider the \texttt{lung} dataset included in the \texttt{survival} package in \texttt{R}, see \cite{R-survival}.   The lung dataset originates from the North Central Cancer Treatment Group study and includes 228 patients, 90 female and 138 male, with advanced lung cancer.  One of the variables of interest is overall survival, measured in days from the initial diagnosis to lost to follow-up or death.  

If we want to compare the survivals for females and males, a first approach is to plot the Kaplan-Meier estimators of the survivals, for both groups.

Formally, let $T_i$, $i=1,\ldots,n$ be independent and identically distributed death times with common survival function $S_T$, and let $C_i$, $i=1,\ldots,n$, be independent and identically distributed censoring times with common survival function $H_C(t)$. It is assumed that failure and censoring times are independent. The dataset consists of bivariate random vectors $(X_i,\delta_i^X)$, where $X_i = T_i \wedge C_i$, where $\wedge$ denotes the minimum, with common distribution function $F$, and $\delta_i^X = I(T_i \le  C_i)$ indicates whether $X_i$ is censored, being $I(\cdot)$ the indicator function. One of the main issues in this context is to provide information about $S_T$ from $(X_i,\delta_i^X)$, $i=1,\ldots,n$. Denoting by $t_1,\ldots, t_k$ the observed death times (non censored observations), the Kaplan-Meier or product limit estimator of the survival function $S_T$ is 
$$
\hat S_{T,n} (t) = \prod_{t_j\le t}\left(1 -  \frac{d_j}{n_j}\right),
$$ 
where $d_j$ is the number of deaths at time $t_j$ and $n_j$ the number of survivors before time $t_j$.

A key result for inferential purposes is the weak convergence of the Kaplan-Meier estimator.  Under the previous notation, \cite{BreslowCrowley:1979} (see also \cite{Gillespie.etal:1979}) proves that for $\tau < + \infty$, such that $S_T(\tau) > 0$, the empirical process
\[
\left\{ n^{1/2}\left( \hat S_{n,T}(t) - S_T(t)\right), t\in (0,\tau)\right\},
\]
converges weakly to a Gaussian process, $G_T$ with mean 0 and covariance matrix given by
\[
cov(G_T(s), G_T(t))=S_T(s)S_T(t)a(\min(s,t))
\]
where 
\[
a(x)=\int_0^x \frac{-dS_T(t)}{(S_T(t))^2 H_C(t)},
\]
under the conditions $S_T$ and $H_C$ being continuous.

We now return to the lung cancer dataset introduced earlier. Figure~\ref{fig1} displays the Kaplan-Meier estimators for both groups,  suggesting that female patients exhibit a higher survival probability across all time points. The traditional methods only give information on equality or differences of the survival function and therefore, it would be more informative to test whether one survival function dominates the other one against the alternative hypothesis that there is, at least, one crossing point between the two survival functions.

\begin{figure}[h!]
    \centering
    \includegraphics[width=0.6\textwidth]{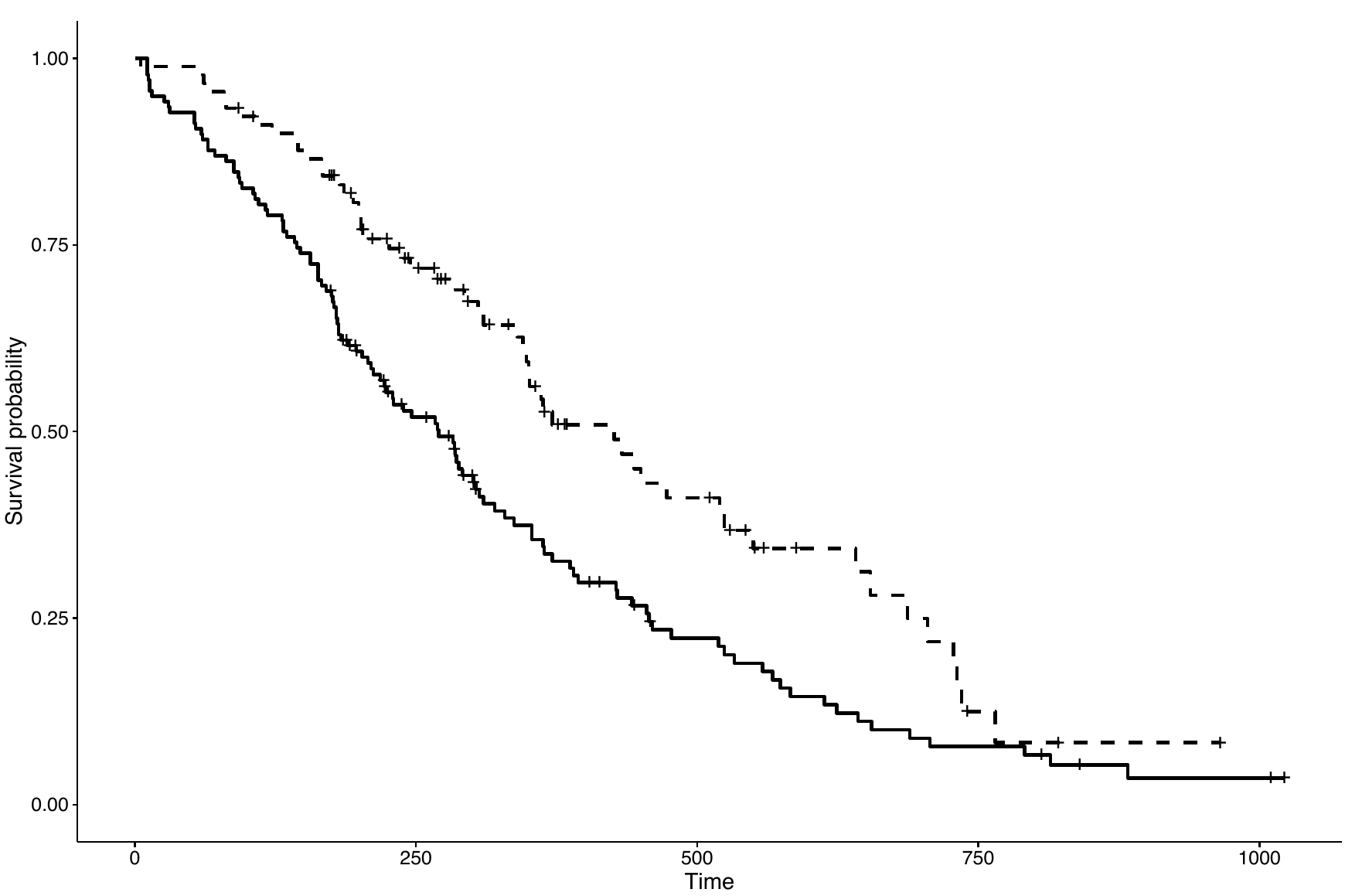}
    \caption{Survival plots for females (dashes) and males (solid).}
    \label{fig1}
\end{figure}

When one survival function lies below another, we use the concept of stochastic ordering. Formally, we say that $T$ is smaller than $U$ in the stochastic, or first stochastic dominance, order, denoted as $T \leq_{st} U$, if $P(T > t) \leq P(U > t)$ for all $t$, meaning that the survival probability of $T$ is always lower than that of $U$. Determining whether one distribution consistently exhibits better survival characteristics than another is essential. Whether this problem has been addressed by Barret and Donald (2003) in the context of stochastic dominance or ordering for non-censored data, there are limited statistical tools available in the case of censored data. 

Following the approach by \cite{BarretDonald:2003}, the main purpose of this paper is to provide statistical methods for testing the ordering of two survival functions by considering Kolmogorov-Smirnov type test in the presence of right-censored data based on the supremum of the difference between the two Kaplan-Meier estimators.

To fix notation, we consider another set $U_i$, $i=1,\ldots,m$ of independent and identically distributed death times with common survival function $S_U$, let $D_i$, $i=1,\ldots,m$, be independent and identically distributed censoring times with common survival function $H_D(t)$. Again, it is assumed that failure and censoring times are independent. The second dataset consists of bivariate random vectors $(Y_i,\delta_i^Y)$, where $Y_i = U_i \wedge D_i$ and $\delta_i^Y = I(U_i \le  D_i)$ indicates whether $Y_i$ is censored or not. Let us denote by $\hat S_{U,m}$ the corresponding Kaplan-Meier estimator of the survival curve $S_U$. We assume that these observations are independent of the previous observations. 

Under the previous notation our main objective is to test the null hypothesis 
\[
H_0: S_T(t)\le S_U(t),\text{ for all }t\in(0,\tau)
\]
against the alternative hypothesis
\[
H_1: S_T(t)> S_U(t),\text{ for some }t\in(0,\tau),
\]
using a test statistic based on the supremum of the difference of the Kaplan-Meier estimators. 

More precisely, we consider the test statistic 
\[
\Delta_{n,m}=\left(\frac{nm}{n+m}\right)^{1/2}\sup_{t \in (0,\tau)}\left\{\hat S_{T,n}(t) - \hat S_{U,m}(t)\right\}. 
\] 

According to this, the null hypothesis is rejected if $\Delta_{n,m}> \delta$ where the critical value $\delta$ would be determined in terms of the distribution of $\Delta_{n,m}$. However, it is not feasible to obtain the exact distribution of
such statistic, and we would rather use the asymptotic distribution. To derive the asymptotic properties we will assume the following assumptions: 

A1: The value $\tau$ satisfies that $S_T(\tau),S_U(\tau) > 0$.

A2: When $n,m\rightarrow +\infty$ then $n/(n+m) \rightarrow \lambda$, with $0< \lambda <1$.

Next, we provide an asymptotic upper bound under $H_0$ for $P(\Delta_{n,m}>\delta)$ in terms of its asymptotic distribution. 

\begin{The}\label{the1} Following the previous notation and assumptions A1 and A2, we get, under $H_0$, that 
\begin{equation}\label{bound-pvalue1}
\lim_{n,m\rightarrow +\infty}P\left(\Delta_{n,m} > \delta \right)\le P\left(\sup_{t\in(0,\tau)} G_{T,U}(t)> \delta\right),
\end{equation}
where $G_{T,U}$ is Gaussian process with 0 mean and covariance matrix given by 
\[
cov(G_{T,U}(s),G_{T,U}(t))=(1-\lambda) cov(G_T(s), G_T(t)) +\lambda cov(G_U(s), G_U(t)).
\]
\end{The}

\begin{proof}As noticed previously,  the empirical processes $n^{1/2}\left( \hat S_{n,T} - S_T\right)$ and $m^{1/2}\left( \hat S_{m,U} - S_U\right)$ converge weakly to Gaussian processes $G_T$ and $G_U$, respectively. The independence of the samples implies that
\[
\left(n^{1/2}\left( \hat S_{n,T} - S_T\right),
m^{1/2}\left( \hat S_{m,U} - S_U\right)\right)
\]
converges weakly to $(G_T,G_U)$. Now, by the continuity mapping theorem we have that 
\[
\Delta_{n,m}^*= \left(\frac{nm}{n+m}\right)^{1/2}\sup_{t \in (0,\tau)}\left\{\hat S_{T,n}(t)- \hat S_{U,m}(t)- (S_T(t)-S_U(t))\right\}
\]
converges weakly to $\sup_{t\in (0,\tau)}\left\{G_{T,U}(t) \right\}$ where $G_{T,U}$ is a Gaussian process with 0 mean and covariance matrix given by 
\[
cov(G_{T,U}(s),G_{T,U}(t))=(1-\lambda) cov(G_T(s), G_T(t)) +\lambda cov(G_U(s), G_U(t)).
\]

Now, under $H_0$ we get that $\Delta_{n,m} \le \Delta_{n,m}^* (a.s)$ and therefore $\Delta_{n,m} \le_{st} \Delta_{n,m}^*$, see Theorem 1.A.1 in \cite{ShakedShanthikumar:2007}. As a consequence, we get
\[
\lim_{n,m\rightarrow +\infty}P(\Delta_{n,m}>\delta) \le \lim_{n,m\rightarrow +\infty}P(\Delta_{n,m}^* >\delta) = P\left(\sup_{t\in(0,\tau)} G_{T,U}(t)> \delta\right).
\]
\end{proof}

This test is consistent as can be seen next. 

\begin{Pro}\label{prop1} Under the conditions of previous theorem and  under $H_1$,  it holds that 
\[
\lim_{n,m\rightarrow +\infty}P\left(\Delta_{n,m} > \delta \right)=1,
\]
for any $\delta \in \mathbb R$.
\end{Pro}

\begin{proof}Under $H_1$ there exists a $t_0 \in (0, \tau)$ such that $ S_T(t_0) - S_U(t_0) >0$, and therefore 
\[
 \sup_{t\in(0,\tau)} \left\{S_T(t) - S_U(t) \right\} >0.
\]

It is not difficult to see that 
\[
 \lim_{n,m\rightarrow +\infty} \sup_{t\in(0,\tau)} \left\{
\hat S_{T,n}(t) - \hat S_{U,m}(t) \right\} =  \sup_{t\in(0,\tau)} \left\{ S_T(t_0) - S_U(t_0) \right\} >0, \text{ a.s.}, 
 \]
and clearly 
\[
 \lim_{n,m\rightarrow +\infty} \left(\frac{nm}{n+m}\right)^{1/2} \sup_{t\in(0,\tau)} \left\{\hat S_{T,n}(t) - \hat S_{U,m}(t) \right\}=  +\infty, \text{ a.s.}. 
 \]
 
 Now the result follows observing that 
\[
\lim_{n,m\rightarrow +\infty} \inf P\left(\Delta_{n,m} > \delta \right) \ge P(\lim_{n,m\rightarrow +\infty} \inf\{\Delta_{n,m} > \delta\}) =1,
\]
for any $\delta \in \mathbb R$. 
\end{proof}

We have introduced the theoretical framework of the proposed test and its motivation compared to traditional methods. Next, we will apply this methodology to real-world datasets and simulations to assess its performance in practical scenarios.

\section{Implementation and application to some datasets}

To provide the upper bound \eqref{bound-pvalue1} for the $p$-value we need to compute the probability for the supremum of a Gaussian process.  A common computational method is to approximate the probability via Monte Carlo simulation. The idea is to generate  $N$  samples of the Gaussian process over a discretized grid,  then compute the supremum for each sample and estimate the probability using empirical frequency.

However, following \cite{BelzunceMartinez:2023}, we propose a more efficient method using the \texttt{mvtnorm} package in \texttt{R}, see \cite{R-mvtnorm}.  Instead of relying exclusively on Monte Carlo simulations, this approach reduces the computational burden associated with the covariance matrix calculation, making the process faster.

Given a discretized  grid $0 = \tau_0 < \tau_1 <\cdots < \tau_m$ on $(0,\tau)$ the Monte Carlo method provides an approximation of $P\left(\sup_{t\in(0,\tau)} G_{T,U}(t)> \delta\right)$, where $\delta$ is the value of $\Delta_{n,m}$ at a given sample, by an estimation of $P(max(G_{T,U}(\tau_0), \ldots, G_{T,U}(\tau_m)) > \delta)$ based on the simulations.  This probability can be computed as
\[
P(max(G_{T,U}(\tau_0), \ldots, G_{T,U}(\tau_m)) > \delta) = 1 - P(G_{T,U}(\tau_0)\le \delta,  \ldots, G_{T,U}(\tau_m))\le \delta).
\]

Given that $(G_{T,U}(\tau_0), \ldots, G_{T,U}(\tau_m))$ follows a multivariate normal distribution, this expression can be efficiently computed in \texttt{R} using the multivariate normal distribution function from the \texttt{mvtnorm} package.

Additionally,  to construct the covariance matrix we require theoretical values of the survival functions $S_T$, $S_U $, and their corresponding densities,  and the survival functions $H_C$,  and $H_D$.  Since these values are unknown, we use empirical values.  The survival functions are replaced by the corresponding Kaplan-Meier estimators and the distribution functions are replaced by their empirical counterpart.  For density estimation, we apply the Foldes-Rejtó-Winter smoothed density estimator, which incorporates plug-in bandwith selection, as proposed by \cite{Jacome.etal:2008} and \cite{CaoLopez:2017}.  This estimation is implemented in \texttt{R} using the  \texttt{survPresmooth} package, see \cite{R-survPresmooth}. 

Finally, to asses assumption A1 we will take $\tau=\min(\max(X_i's), \max(Y_j's))$.

Next, we provide applications of the previous test to some datasets, starting with the lung cancer dataset introduced in Section 2.  Here,  we aim to assess whether the survival function of male patients lies consistently below that of female patients, as Figure~\ref{fig1} suggests.

Applying our test, we obtain $\Delta_{n,m}=-0.0375$ and an upper bound for the $p$-value of 0.9955, indicating no statistical evidence to reject the null hypothesis that female survival dominates male survival.

To determine whether this dominance is strict or the two survival functions are equal, we apply some classical tests to detect differences among two survival functions.  The results are presented in Table \ref{table1}. Since all tests yield extremely low $p$-values, we conclude that female survival strictly dominates male survival.  

\begin{table}[h]
\centering
{	\begin{tabular}{crrcrrc}
	   Test &  Test statistic & \( p\text{-value} \) \\ \hline
	Log-Rank & 10.3267 & 0.0013 \\
	Gehan & 12.4721 &  0.0004 \\
	Tarone--Ware & 12.4555 & 0.0004 \\
	Peto--Peto & 12.7078 & 0.0003 \\
	Modified Peto--Peto  & 12.7091 & 0.0003 \\ \hline
	\end{tabular}}
	\label{table1}
	\caption{Test statistics and $p$-values of some classical tests for the \texttt{lung} dataset.}
\end{table}

Next, we consider times to infection of kidney dialysis patients. This dataset records the time to first infection for patients undergoing kidney dialysis, 43 patients utilized a surgically placed catheter, Group 1, and 76 patients utilized a percutaneous placement of their catheter, Group 2. Since not all patients experience an infection during the observation period, the data includes right-censored observations. The time to infection is measured in days. The example is taken from \cite{KleinMoeschberger:2003}, Example 1.4, page 6.

In this case we compare times to infection for Group 1 and Group 2.  Figure~\ref{fig2} presents the Kaplan-Meier survival curves for both groups, suggesting a crossing point between them.

\begin{figure}[h!]
    \centering
    \includegraphics[width=0.6\textwidth]{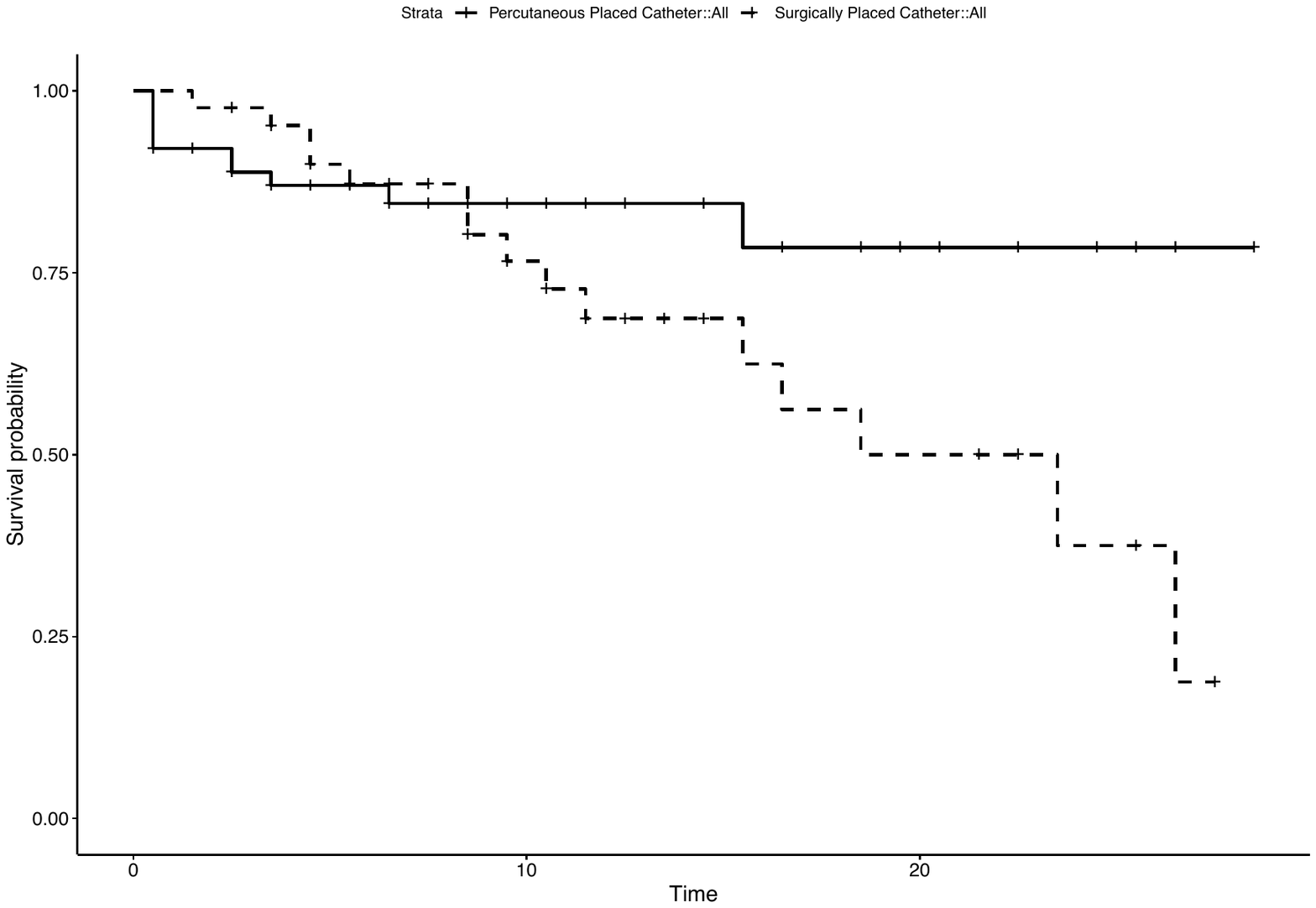}
    \caption{Survival plots of time to infection for percutaneous placement of catheter (dashes) and surgically placed catheter (solid).}
    \label{fig2}
\end{figure}

Applying our proposed test, we obtain $\Delta_{n,m}=3.1305$ and an upper bound for the $p$-value of 0.0002, which provides statistical evidence that at least one crossing point exists between the two survival functions, confirming what is observed in the graphical representation.

On the other hand, when applying traditional tests for comparing survival functions, we obtain Table \ref{table2}, see Table 7.3 in \cite{KleinMoeschberger:2003}, page 210: 

\begin{table}[h]
\centering
{	\begin{tabular}{crrcrrc}
		Test &  Test statistic & \( p\text{-value} \) \\ \hline
		Log-Rank & 2.5295 & 0.11 \\
        Gehan  & 0.0020 & 0.96 \\
        Tarone--Ware & 0.4027 & 0.52 \\
        Peto--Peto & 1.3991 & 0.23 \\
        Modified Peto--Peto & 1.2759 & 0.25 \\ 
        \hline
	\end{tabular}}
	\label{table2}
	\caption{Test statistics and $p$-values of some classical tests for the kidney \break dialysis dataset.}
\end{table}

None of these tests detect a significant difference between the two survival functions, which contradicts the results of our supremum-based Kaplan-Meier difference test.  This suggests that traditional methods may lack the sensitivity to detect structural changes in survival relationships, particularly in the presence of crossing survival curves.

\section{Simulation studies}

To show the performance of our test in different scenarios we carry out Monte Carlo experiments for small and large samples. The simulation studies are performed in several situations where the dominance among the two survivals either holds or does not hold, for different samples sizes and different rates of censoring.

We have considered gamma distributed survival times and exponentially distributed censoring times.  Let us describe the different cases that we have considered. 

Case 1: In this case $T$ follows a gamma distribution with shape parameter 2 and scale parameter 1, $T\sim Gamma(2,1)$, and $U\sim Gamma(3,1)$. It is known that in this case the survival function of $U$ dominates the survival function of $T$. 

Case 2: In this case $T\sim Gamma(2,1)$ and $U\sim Gamma(2.2,1)$. It is known that in this case the survival function of $U$ dominates the survival function of $T$, but they are very close. 

Case 3: In this case $T\sim Gamma(3,5)$ and $U\sim Gamma(6,2)$. It is known that in this case the survival functions cross at one point. 

Case 4: In this case $T\sim Gamma(2,2)$ and $U\sim Gamma(3,1)$. It is known that in this case the survival functions cross at one point, but it is difficult to detect the crossing point.

In Figure~\ref{fig3} we have plot the different cases, where we observe ordered survival functions in cases 1 and 2 and a crossing point in cases 3 and 4.

\begin{figure}[h!]
    \centering
    \includegraphics[width=1.0\textwidth]{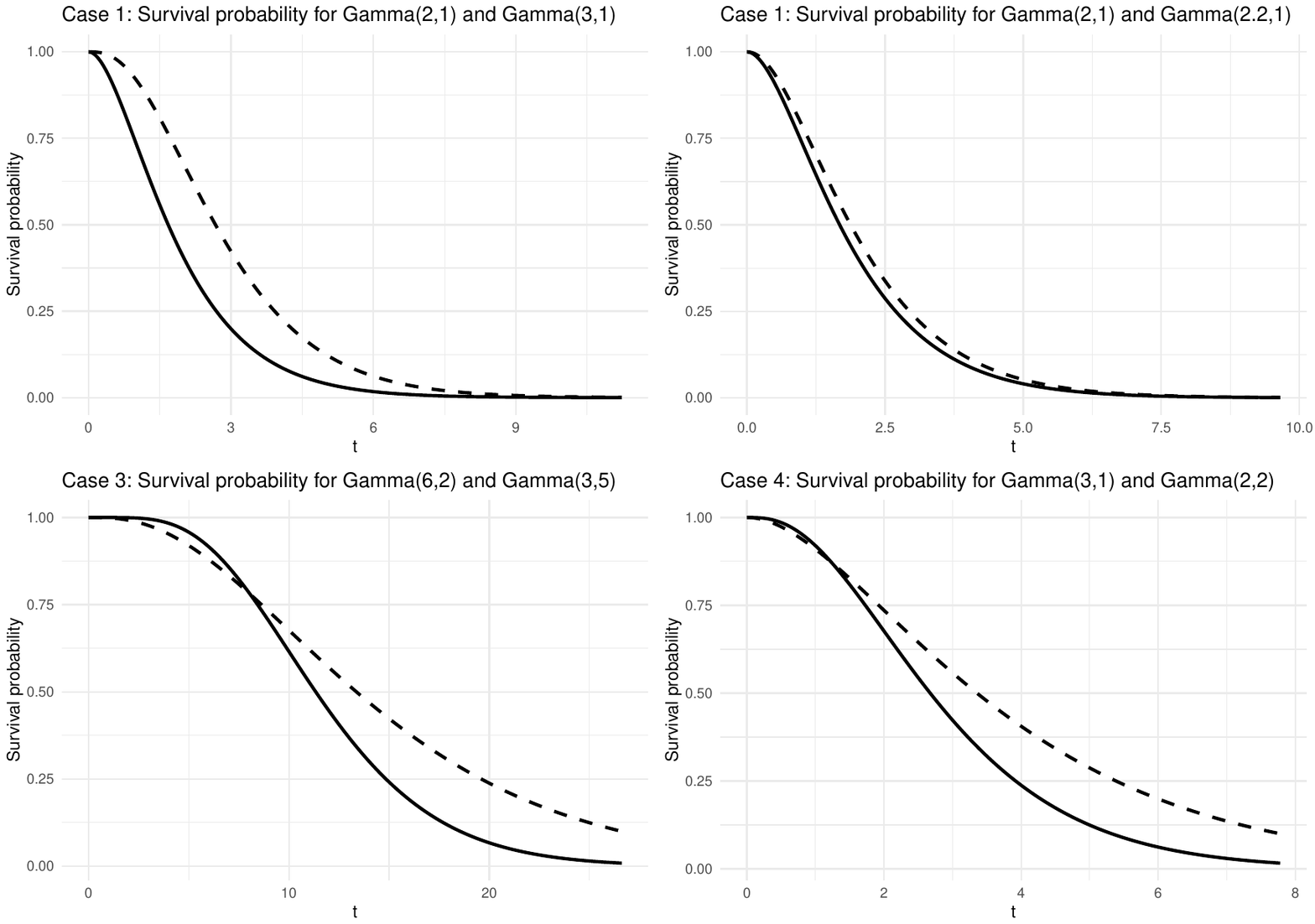}
    \caption{Survival plots for $T$ (dashes) and $U$ (solid), for cases 1--4.}
    \label{fig3}
\end{figure}

For each case we have considered exponentially distributed censor times with different shape parameter, $\lambda$. In each case, we have selected two different $\lambda$'s  to get approximately a rate of 20\% and a 50\% of censored observations. To get a 20\% and a 50\% rate of censoring we have considered $\lambda$ equal to $-\log(0.9)/q_{20}$ and $\log(2)/q_{50}$, respectively, where $q_{20}$ and $q_{50}$ are the 20\% and 50\% quantile of the corresponding gamma distribution.  

We have performed 1000 Monte Carlo replications for each case with different sample sizes, $n=50, 100, 150, 200, 250$ and $500$, in which the rejection rates of the null hypothesis have been computed for the two conventional significance levels of $0.05$ and $0.01$. The number of points of the grid is $m=100$ in every replication.

Table \ref{montecarlo} summarizes the performance of the proposed test under four scenarios, Cases 1--4, with varying sample sizes, $n=50, 100, 150, 200, 250, 500$ and censoring levels,  20\% vs 50\%.

Some key observations are the following: 

\begin{enumerate}

\item When one survival function truly dominates the other, Cases 1 and 2, the test almost never rejects the null hypothesis. Even with small samples, the rejection rate is essentially 0 at both $\alpha=0.05$ and $0.01$, indicating no false alarms when the survival functions are ordered.

\item In scenarios where the survival functions cross, Cases 3 and 4, the test's power to reject the null hypothesis increases sharply with larger sample sizes. For small samples, e.g. $n=50$, the rejection rates are modest, but with larger samples they increase,  exceeding 90\% by $n=200$ and approaching 100\% at $n=250$. This shows the test is very sensitive to crossings given sufficient data.

\item Heavier censoring reduces the test's power to detect differences. At a given sample size, a 50\% censoring rate yields lower rejection rates than 20\% censoring. For example, in Case 3 with $n$=100, the rate of rejection is 86.3\% with 20\% censoring versus 64\% with 50\% censoring at $\alpha=0.05$. Nonetheless, as sample size grows, even with 50\% censoring the power eventually becomes high, reaching 99\% at $n$=500 in crossing cases.

\end{enumerate}

These results demonstrate that the proposed test is reliable for confirming the order of survival functions and highly effective at flagging crossings, provided the data are sufficient. The test maintains a low false-positive rate when survival curves are truly ordered, giving confidence that a non-significant result indeed suggests no crossing. Conversely, a significant result from this test is strong evidence of at least one crossing point between survival curves. Sufficient sample size is crucial for the test to detect crossings, especially if censoring is heavy. In practice, researchers should plan for larger samples and/or seek to reduce censoring to ensure the test has high power to uncover crossing survival patterns . This sensitivity to crossings is a key advantage over traditional survival comparison tests (like log-rank), which often fail to detect any difference when survival curves intersect. Such improved detection can lead to better insights in studies where survival functions may cross, ensuring that meaningful survival differences are not overlooked.

\begin{table}[h]
\centering
{	\resizebox{!}{0.072\textheight}{\begin{tabular}{cccccccccccc}
		$n$ & \multicolumn{2}{c}{Case 1} & &\multicolumn{2}{c}{Case 2} & &\multicolumn{2}{c}{Case 3} & &\multicolumn{2}{c}{Case 4} \\ \hline
		 & 20\% & 50\% & &20\% & 50\% & &20\% & 50\% & &20\% & 50\% \\ 
		50  & 0 (0) & 0 (0) && 0.010 (0.002) & 0.012 (0.003) && 0.541 (0.274) & 0.359 (0.144) && 0.505 (0.246) & 0.309 (0.199)\\ 
		100 & 0 (0) & 0 (0) && 0.008 (0.002) & 0.010 (0.002) && 0.863 (0.629) & 0.640 (0.356) && 0.770 (0.499) & 0.537 (0.266)\\ 
		150 & 0 (0) & 0 (0) && 0.003 (0.001) & 0.004 (0.001) && 0.966 (0.848) & 0.822 (0.582) && 0.907 (0.756) & 0.724 (0.474)\\ 
		200 & 0 (0) & 0 (0) && 0.002 (0.000) & 0.003 (0.000) && 0.990 (0.952) & 0.933 (0.774) && 0.962 (0.867) & 0.838 (0.630)\\ 
        250 & 0 (0) & 0 (0) && 0.002 (0.000) & 0.002 (0.000) && 1.000 (0.977) & 0.976 (0.889) && 0.991 (0.935) & 0.921 (0.761) \\
        500 & 0 (0) & 0 (0) && 0.000 (0.000) & 0.002 (0.000) && 1.000 (1.000) & 0.998 (0.994) && 1.000 (1.000)  & 0.996 (0.981) \\
\hline
	\end{tabular}}}
	\caption{Rejection rates for the null hypothesis for $\alpha = 0.05$ ($0.01$)}
	\label{montecarlo}
\end{table}

\section{Discussion}

In this paper, we introduced a new test for assessing the dominance of survival functions in the presence of right-censored data. Unlike traditional methods, which primarily detect whether survival curves differ, our approach provides additional insight by distinguishing between strict ordering and potential crossings of the survival functions. This feature is particularly relevant in applications where survival curves exhibit intersections, which often occur in medical and reliability studies.

Our proposed test, based on the supremum of the difference of Kaplan-Meier estimators, offers several advantages over traditional methods such as the log-rank test:
\begin{enumerate}

\item Greater sensitivity to survival function crossings: Traditional tests, like the log-rank test, focus on global differences between survival curves and may fail to detect structural changes, such as dominance relationships or crossings. Our approach explicitly accounts for these aspects, improving the interpretability of survival comparisons.

\item Improved detection of stochastic dominance: Our test is specifically designed to assess whether one survival function consistently lies above another, providing a more informative alternative to conventional hypothesis testing frameworks.

\item Computational efficiency and implementation in \texttt{R}: The test is easy to compute using standard statistical software. The estimation of $p$-values can be efficiently performed using multivariate normal approximations, making it practical for large datasets.

\end{enumerate}

Despite its strengths, the proposed method has certain limitations that suggest promising avenues for future research:
\begin{enumerate}

\item Performance in small sample sizes: While our test has strong asymptotic properties, its performance in small samples could be further refined by improving the estimation of the asymptotic distribution. Future research could explore bootstrap-based refinements or finite-sample adjustments to enhance the test?s accuracy in small-sample scenarios.

\item Adaptation for paired samples: A relevant extension would be the development of a similar test for paired survival data, where observations are correlated (e.g., matched case-control studies or twin studies).

\item Handling time-dependent covariates: In real-world applications, survival probabilities often depend on time-varying covariates. Extending our approach to incorporate time-dependent covariates could provide additional insights into how survival dominance evolves over time.

\end{enumerate}

Overall, our test represents a novel and powerful alternative for comparing survival functions, particularly when the interest lies in detecting dominance or structural changes in survival relationships. By incorporating supremum-based differences in Kaplan-Meier estimators, our approach overcomes the limitations of traditional methods and offers a more nuanced view of survival function comparisons. We believe this method will be particularly useful in medical, epidemiological, and reliability studies, where detecting dominance and survival function crossings is crucial for decision-making.

\section*{Acknowledgement}
This work has been funded by Agencia Estatal de Investigación (AEI) del Ministerio de Ciencia e Innovación and cofunded by Fondo Europeo de Desarrollo Regional (FEDER), project Comparaciones Estocásticas Bajo Estructuras de Dependencia con Aplicaciones a Fiabilidad y Finanzas. Teoría e Inferencia (PID2022-137396NB-I00).

\bibliographystyle{plain}
\bibliography{paper-ref}
\end{document}